\documentclass[11pt]{article}
\usepackage[a4paper,margin=1in]{geometry}
\usepackage{amsmath,amssymb,amsthm}
\usepackage{booktabs}
\usepackage{enumitem}
\usepackage{hyperref}
\usepackage{xspace}
\usepackage{algorithm}
\usepackage{algpseudocode}
\usepackage{tabularx}

\newcommand{\ARACE}{AR-ACE\xspace}
\newcommand{\Adv}{\mathsf{Adv}}
\newcommand{\negl}{\mathsf{negl}}
\newcommand{\REV}{\mathsf{REV}}
\newcommand{\Ctx}{\mathsf{Ctx}}
\newcommand{\ObjHash}{\mathsf{objHash}}
\newcommand{\Attest}{\mathsf{Attest}}
\newcommand{\domain}{\mathsf{domain}}

\newtheorem{theorem}{Theorem}[section]

\newtheorem{proposition}{Proposition}[section]

\title{\textbf{AR-ACE: ACE-GF-based Attestation Relay for PQC---\\
Lightweight Mempool Propagation Without On-Path Proofs}}
\author{
  \texttt{Jian Sheng Wang} \\
  \texttt{jason@acechain.io} \\
  Montreal, Canada
}
\date{November 24, 2026}

\begin{document}
\maketitle

\begin{abstract}
In post-quantum blockchain settings, objects that require validity proofs (e.g., blob roots, execution-layer or consensus-layer signature aggregates) must be broadcast through mempool and relay networks. Recursive STARKs have been proposed to aggregate such proofs so that each node forwards one proof per tick plus objects without proofs, capping per-node proof bandwidth at roughly 128\,KB $\times$ degree per tick. We observe that propagation does not require \emph{validity proofs} on the path; relays need only a lightweight assurance that an object is eligible for relay. We present \textbf{AR-ACE} (ACE-GF-based Attestation Relay for PQC), in which relay nodes forward \emph{objects plus compact attestations} (e.g., identity-bound signatures or commitments) and do not generate, hold, or forward any full validity proof. Only the builder (or final verifier) performs a single aggregated validity proof over the set of objects it includes. In this \emph{proof-off-path} design no validity proof travels on the propagation path, so \emph{proof-related} relay bandwidth is zero. When attestation keys are derived with ACE-GF, they share an identity root with on-chain authorization keys and can use PQC algorithms. We specify a protocol model, state design goals and security considerations, define security games, and provide a structural bandwidth comparison with recursive-STARK-based propagation.
\end{abstract}


\section{Introduction}
\label{sec:intro}

\subsection{The Bandwidth Problem in Post-Quantum Mempools}
\label{sec:intro:problem}
In distributed block construction, a large number of objects (such as blob roots with erasure-coding correctness, execution-layer transactions, or consensus-layer signature aggregates) must be broadcast so that builders can discover and include them. In a post-quantum setting, elliptic-curve SNARKs may be considered unsuitable; STARKs~\cite{stark} are the usual alternative for proving validity of these objects. A single STARK proof, however, is on the order of 128\,KB even in size-optimized implementations. If every object were propagated together with its full STARK proof, the bandwidth demand on intermediate mempool nodes and on the builder would scale linearly with the number of objects, quickly becoming prohibitive.

Buterin~\cite{vbuterin-recursive-stark-mempool} proposed addressing this by \emph{recursive STARKs}: each mempool node periodically produces one recursive STARK attesting to the validity of all objects it currently holds, and forwards to each peer this single proof plus any objects (without proofs) not yet sent to that peer. The key result is that per-node proof bandwidth becomes \emph{constant} in the number of objects: roughly one STARK per tick per peer (e.g., 128\,KB $\times$ 8 / 0.5\,s $\approx$ 2\,MB/s outbound per node), while object payloads are transmitted as in the status quo. Proof traffic is thus bounded independently of the number of objects.

\subsection{Key Observation: What Must Be Proven on the Path?}
\label{sec:intro:key-obs}
We ask whether validity proofs need to appear on the propagation path at all.

Relay nodes need to decide whether to accept and forward an object. That decision could be based on:
\begin{itemize}[nosep,leftmargin=1.5em]
  \item a \emph{full validity proof} (e.g., STARK) for the object, as in proof-carrying and recursive-STARK propagation; or
  \item a \emph{lightweight credential} that attests to the object's eligibility for relay (e.g., the submitter or an attestation service signs a credential bound to the object hash and domain, under a relay quota), without proving that the object is well-formed or valid on the path.
\end{itemize}
If we adopt the second view, then \emph{validity} is proven only when the builder includes the object, for example by running one aggregated STARK (or equivalent) over the set of objects it selects. Propagation need only carry the object plus a small attestation (signature, commitment, or bond reference). We call this \emph{proof-off-path} propagation: proofs stay off the relay path; only the builder sees or produces them.

The principle is general. When a heavy cryptographic artifact (proof, signature) must be verified at a final stage, it need not be carried on every intermediate hop. Instead, intermediate nodes verify only a lightweight credential sufficient for relay eligibility, and the heavy verification is deferred to the endpoint (here, the builder).

\subsection{Contributions}
\label{sec:intro:contrib}
\begin{enumerate}[nosep,leftmargin=1.5em]
  \item We introduce \ARACE (ACE-GF-based Attestation Relay for PQC), a propagation design in which relay nodes forward objects plus compact attestations and do not generate, hold, or forward full validity proofs.
  \item We specify the proof-off-path protocol model: roles (submitters, relay nodes, builder), the attestation credential, and the rule that only the builder performs aggregated validity verification.
  \item We state design goals (bandwidth efficiency, spam resistance, binding of attestation to object and identity), define security games (attestation forgery, replay, admission bypass), and discuss security considerations for attestation-based relay.
  \item We provide a structural bandwidth comparison showing that \ARACE carries no proof traffic on the propagation path, while remaining compatible with a single aggregated proof at the builder.
  \item We describe what ACE-GF adds as the source of attestation keys: one identity root for both on-chain authorization and relay attestation, context isolation, and PQC-ready derivation.
\end{enumerate}

\subsection{Relation to ACE-GF}
\label{sec:intro:relation}
\ARACE does not \emph{require} ACE-GF: the proof-off-path architecture can be instantiated with any mechanism that supplies a compact attestation (e.g., BLS signature, Ed25519, bond reference). When instantiated with ACE-GF-derived keys~\cite{arxiv-2511-20505} for attestation, the same identity root that backs on-chain authorization can derive a dedicated attestation key for mempool relay, isolated by context and available with PQC algorithms. This paper presents \ARACE as a general design and describes the ACE-GF instantiation for deployments that want one identity root and post-quantum consistency.

\subsection{Paper Organization}
Section~\ref{sec:related} surveys related work. Section~\ref{sec:prelim} gives notation and assumptions. Section~\ref{sec:problem} states the problem and design goals. Section~\ref{sec:design} presents the \ARACE design. Section~\ref{sec:security} discusses security, including the threat model and security games. Section~\ref{sec:bandwidth} compares bandwidth and outlines practical optimization directions. Section~\ref{sec:acegf} summarizes ACE-GF advantages. Section~\ref{sec:conclusion} concludes.

\section{Related Work}
\label{sec:related}

\paragraph{Recursive STARK mempool.}
Buterin~\cite{vbuterin-recursive-stark-mempool} proposed recursive STARKs~\cite{stark} so that each node produces one proof per tick attesting to the validity of all objects it knows; objects are sent without proofs. This technique builds on recursive proof composition~\cite{halo,nova}. It caps per-node proof bandwidth but still requires every node to produce and forward STARKs. Our approach removes proofs from the path; the builder alone performs (or verifies) an aggregated proof.

\paragraph{Identity-centric authorization.}
One approach to post-quantum blockchain authorization, ZK-ACE~\cite{arxiv-2603-07974}, proves authorization with zero-knowledge proofs bound to a registered identity commitment instead of a signature public key. \ARACE applies the same principle to the relay layer: we do not compress the heavy artifact (proof); we change \emph{what} is verified on the path (attestation of eligibility vs.\ validity proof) and defer the heavy verification to the builder.

\paragraph{Data availability and blobs.}
EIP-4844~\cite{eip4844} introduces blob transactions and data-availability sampling (DAS). Propagation of blob roots and their validity (e.g., erasure-coding correctness) is one use case for both recursive STARK propagation and \ARACE. Our design is agnostic to the exact validity predicate; we only require that the builder can eventually verify inclusion in a single aggregated proof.

\paragraph{P2P gossip and relay networks.}
Blockchain mempool propagation commonly relies on structured gossip protocols such as GossipSub~\cite{gossipsub}, which provides attack-resilient message dissemination over unstructured overlays, building on DHT foundations like Kademlia~\cite{kademlia}. \ARACE is agnostic to the underlying gossip topology; it only requires that relay nodes can verify attestations before forwarding.

\paragraph{Account abstraction and relayers.}
ERC-4337~\cite{erc4337} and relayer/bundler architectures separate the entity that signs or authorizes from the entity that submits. \ARACE is compatible with such separation: the attestation can be produced by the submitter or by a relayer holding a dedicated attestation key (e.g., derived from ACE-GF under a relay-specific context).

\section{Preliminaries and Assumptions}
\label{sec:prelim}

\subsection{Notation}
\label{sec:prelim:notation}
\begin{itemize}[nosep,leftmargin=1.5em]
  \item $\lambda$: security parameter.
  \item $\ObjHash$: hash of an object (blob root, transaction, or aggregate) under a collision-resistant hash.
  \item $\Attest$: attestation credential (e.g., signature over $\ObjHash$ or commitment binding object and identity).
  \item $\domain$: domain separator (chain, mempool instance, or application).
  \item $\REV$: Root Entropy Value (ACE-GF identity root); when ACE-GF is used, attestation keys are derived from $\REV$ under a relay-specific $\Ctx$.
  \item Relay path: the sequence of nodes from submitter(s) to the builder. \emph{On-path} means on this path; \emph{off-path} means not on this path (e.g., only at the builder).
\end{itemize}

\subsection{Assumed ACE-GF Interface (When Used)}
\label{sec:prelim:acegf}
When \ARACE is instantiated with ACE-GF~\cite{arxiv-2511-20505}, we assume the following black-box interface: deterministic key derivation $\mathsf{Key} \leftarrow \mathsf{Derive}(\REV, \Ctx)$ with context isolation. For attestation, we use a dedicated context, e.g., $\Ctx_{\mathsf{relay}} = (\mathsf{AlgID}, \mathsf{MEMPOOL\text{-}ATTEST}, \mathsf{domain})$, so that the attestation key is isolated from signing, encryption, or on-chain authorization keys. No persistent storage of $\REV$ is required; the attestation key is derived ephemerally when producing $\Attest$.

\subsection{Network and Role Model}
\label{sec:prelim:network}
\begin{itemize}[nosep,leftmargin=1.5em]
  \item \textbf{Submitters}: Produce objects (e.g., blob roots, transactions) and an attestation $\Attest$ binding themselves (or their identity commitment) to $\ObjHash$ (and optionally $\domain$, nonce). They send $(\mathsf{obj}, \Attest)$ into the mempool.
  \item \textbf{Relay nodes}: Receive $(\mathsf{obj}, \Attest)$ from peers or submitters. They verify that $\Attest$ is valid for $\ObjHash$ (and domain) and that the object is not yet expired or discarded. They do \emph{not} verify full validity (e.g., erasure coding, signature aggregate correctness). They forward $(\mathsf{obj}, \Attest)$ to their peers. They do not generate or forward any STARK or other validity proof.
  \item \textbf{Builder}: Collects $(\mathsf{obj}, \Attest)$ from the network. Selects a set of objects to include. Produces or verifies \emph{one} aggregated validity proof (e.g., recursive STARK) over the included set. Publishes the block (or batch) together with this single proof.
\end{itemize}

\section{Problem Statement and Design Goals}
\label{sec:problem}

\subsection{Bandwidth Problem}
\label{sec:problem:bandwidth}
The bandwidth problem we address is the following. Many objects must be broadcast so that a builder can include them. Each object has an associated \emph{validity} predicate (e.g., erasure coding correct, signatures valid). In a post-quantum setting, validity may be proved with a STARK; a single STARK is $\approx$128\,KB. If every object is propagated with its proof, relay bandwidth grows as $O(\text{\#objects} \times 128\,\text{KB})$. Recursive STARK aggregation~\cite{vbuterin-recursive-stark-mempool} reduces this to $O(1)$ proof traffic per node per tick (one proof per peer per tick), but each node must still produce and send STARKs. Our goal is to remove proof traffic from the relay path while preserving the guarantee that only valid objects can be included by the builder.

\subsection{Design Goals}
\label{sec:problem:goals}
\begin{itemize}[nosep,leftmargin=1.5em]
  \item \textbf{Bandwidth efficiency.} Relay nodes must not send or receive full validity proofs. Only objects and compact attestations (e.g., tens to hundreds of bytes per object) flow on the path.
  \item \textbf{Spam resistance.} Attestation must be costly enough (or rate-limited) so that anonymous spam is impractical. Options include: attestation key tied to identity or stake, small bond per attestation, or rate limits per attestation key.
  \item \textbf{Binding.} $\Attest$ must bind to $\ObjHash$ (and optionally $\domain$) so that attestations cannot be replayed on different objects or domains.
  \item \textbf{Validity at inclusion.} The builder verifies validity of included objects via a single aggregated proof. No object without a valid attestation need be considered; the builder may additionally require that attestations come from eligible keys or commitments.
  \item \textbf{Compatibility.} The design should be compatible with existing validity predicates (erasure coding, signature aggregation) and with a single aggregated proof (e.g., recursive STARK) at the builder.
\end{itemize}

\paragraph{Minimal admission policy model (for this paper).}
\label{sec:problem:admission}
To make the spam-resistance assumptions explicit, we assume a minimal admissibility policy at relay:
(i) only attestations from eligible keys (registered or bonded) are admitted,
(ii) each eligible key has a per-epoch quota, and
(iii) policy violations trigger drop plus temporary denylisting.
This paper does not optimize policy economics; it only requires that admissibility is enforceable
and expensive enough that large-scale spam is not costless.

\section{AR-ACE Design}
\label{sec:design}

\subsection{Attestation Credential}
\label{sec:design:attest}
An attestation $\Attest$ is a compact credential that binds the attester (e.g., identity or public key) to the object hash $\ObjHash$ and optionally to $\domain$ and a nonce or replay token. Concretely, $\Attest$ can be:
\begin{itemize}[nosep,leftmargin=1.5em]
  \item A signature $\sigma$ under key $\mathsf{pk}$ on message $m = \ObjHash \| \domain \| e \| \mathsf{nonce}$, where $e$ is the current quota epoch (or a canonical encoding), using a standard signature scheme with EUF-CMA security~\cite{gmr88}. Verification on the path checks $\mathsf{Verify}(\mathsf{pk}, m, \sigma)$ and optionally that $\mathsf{pk}$ is in an allowed set or derived from a committed identity.
  \item When using ACE-GF: $\mathsf{pk}$ is the public key corresponding to $\mathsf{Derive}(\REV, \Ctx_{\mathsf{relay}})$. The same $\REV$ that backs on-chain identity commitments derives the relay key under a separate context, so compromise of a relay-only key does not expose authorization keys.
\end{itemize}
Size: a signature (e.g., Ed25519~\cite{ed25519} 64 bytes, ML-DSA~\cite{mldsa} level 2 $\approx$2.4\,KB if PQC is used for attestation) plus optional $\mathsf{pk}$ or commitment.

\paragraph{Deployment profiles.}
\textbf{Profile A (Hybrid-performance):} on-path attestations use short classical signatures (e.g., Ed25519~\cite{ed25519} or BLS~\cite{bls}) to minimize bandwidth and relay CPU; PQC security is enforced at inclusion/final verification layers. Relay admission (eligibility, quotas, spam resistance) then holds only against classical adversaries, since a quantum adversary can recover eligible signing keys from their public keys.\\
\textbf{Profile B (Full-PQC):} on-path attestations use PQC signatures (e.g., ML-DSA~\cite{mldsa}), increasing per-object bytes but keeping a fully PQC-consistent relay path.

\subsection{Propagation Protocol}
\label{sec:design:propagation}
\begin{enumerate}[nosep,leftmargin=1.5em]
  \item Submitter computes $\ObjHash = H(\mathsf{obj})$, produces $\Attest$ (signature or commitment) binding $\ObjHash$, $\domain$, and optional nonce. Sends $(\mathsf{obj}, \Attest)$ to one or more relay nodes.
  \item Relay node: on receipt of $(\mathsf{obj}, \Attest)$, verifies $\Attest$ for $\ObjHash$ and $\domain$; checks replay/expiry. Does \emph{not} verify full validity (no STARK, no erasure-coding check). Forwards $(\mathsf{obj}, \Attest)$ to peers that have not yet received this object.
  \item Builder: subscribes to or pulls objects (with attestations) from the network. Selects a set $S$ of objects to include. Optionally discards objects that are no longer valid for inclusion. Produces one aggregated validity proof $\Pi$ for $S$ (e.g., recursive STARK: ``all objects in $S$ satisfy the validity predicate''). Publishes block/batch with $S$ and $\Pi$. Verification of the block checks $\Pi$; no per-object proof is verified on-chain.
\end{enumerate}

\noindent Algorithm~\ref{alg:arace} formalizes the three roles.

\begin{algorithm}[t]
\caption{AR-ACE Propagation Protocol}\label{alg:arace}
\begin{algorithmic}[1]
\Statex \textbf{Parameters:} signature scheme $(\mathsf{KeyGen}, \mathsf{Sign}, \mathsf{Verify})$, hash $H$, domain $\domain$, eligible key set $\mathcal{E}$, per-key quota $Q$
\Statex
\Statex \underline{\textsc{Submitter}($\mathsf{obj}, \mathsf{sk}$)}
\State $\ObjHash \gets H(\mathsf{obj})$
\State $\mathsf{nonce} \gets$ fresh replay token
\State $m \gets \ObjHash \| \domain \| e \| \mathsf{nonce}$ \Comment{$e$: current epoch}
\State $\Attest \gets \mathsf{Sign}(\mathsf{sk}, m)$
\State \textbf{send} $(\mathsf{obj}, \Attest, \mathsf{pk}, \mathsf{nonce})$ to relay peers
\Statex
\Statex \underline{\textsc{Relay}(received $(\mathsf{obj}, \Attest, \mathsf{pk}, \mathsf{nonce})$)}
\State $\ObjHash \gets H(\mathsf{obj})$
\State $m \gets \ObjHash \| \domain \| e \| \mathsf{nonce}$ \Comment{$e$: relay's current epoch}
\If{$\ObjHash$ seen before} \textbf{drop and return} \EndIf
\If{$\mathsf{pk} \notin \mathcal{E}$} \textbf{drop and return} \EndIf
\If{$\mathsf{quota}[\mathsf{pk}] \geq Q$} \textbf{drop and return} \EndIf
\If{$(\mathsf{pk}, \mathsf{nonce})$ seen before} \textbf{drop and return} \EndIf
\If{$\mathsf{Verify}(\mathsf{pk}, m, \Attest) = 0$} \textbf{drop and return} \EndIf
\State $\mathsf{quota}[\mathsf{pk}] \gets \mathsf{quota}[\mathsf{pk}] + 1$
\State record $(\mathsf{pk}, \mathsf{nonce})$ and $\ObjHash$
\State \textbf{forward} $(\mathsf{obj}, \Attest, \mathsf{pk}, \mathsf{nonce})$ to peers not yet served
\Statex
\Statex \underline{\textsc{Builder}(object pool $\mathcal{P}$)}
\State $S \gets$ select objects $\mathsf{obj} \in \mathcal{P}$ by policy such that $\mathsf{Valid}(\mathsf{obj})$ \Comment{per-object check at the builder}
\State $\Pi \gets \mathsf{AggregatedProve}(\{(\mathsf{obj}, \text{validity predicate}) : \mathsf{obj} \in S\})$
\If{$\Pi$ valid} \textbf{publish} block $(S, \Pi)$ \EndIf
\end{algorithmic}
\end{algorithm}

\subsection{What Flows on the Path}
\label{sec:design:flow}
On every link of the relay path:
\begin{itemize}[nosep,leftmargin=1.5em]
  \item \textbf{Transmitted:} $\mathsf{obj}$ (unchanged) and $\Attest$ (constant size). No STARK, no recursive proof.
  \item \textbf{Not transmitted:} Any validity proof. Proofs exist only at the builder (output of aggregation) and in the published block.
\end{itemize}
This is the \emph{proof-off-path} property.

\subsection{ACE-GF Instantiation (Summary)}
\label{sec:design:acegf}
When attestation keys are derived via ACE-GF:
\begin{itemize}[nosep,leftmargin=1.5em]
  \item Use $\Ctx_{\mathsf{relay}} = (\mathsf{AlgID}, \mathsf{MEMPOOL\text{-}ATTEST}, \domain)$ so that relay attestation is in its own stream.
  \item The same $\REV$ (reconstructed from sealed artifact and credential) that backs on-chain identity commitments can derive $\mathsf{sk}_{\mathsf{relay}}$. No separate root, backup, or key management is needed for the relay key.
  \item Optional: use a PQC algorithm in $\Ctx_{\mathsf{relay}}$ for attestation (e.g., ML-DSA) so that the entire path is PQC-consistent; then $\Attest$ is larger but still far smaller than a STARK, and no proof flows on the path.
\end{itemize}

\section{Security Considerations}
\label{sec:security}

\subsection{Threat Model and Security Games}
\label{sec:security:games}
We consider a probabilistic polynomial-time adversary with full visibility into mempool traffic and the ability to submit, replay, and reorder relayed objects. Security is expressed with three games following the code-based game-playing framework~\cite{bellare-rogaway}.

\paragraph{Game $\mathsf{G}_{\mathsf{forge}}$ (Attestation Forgery).}
The challenger exposes public eligibility metadata and oracle access for valid attestations under eligible keys. The adversary wins if it outputs $(\mathsf{obj},\Attest)$ that passes relay verification for an eligible key without querying that message. Advantage is
\[
\Adv^{\mathsf{forge}}_{\ARACE}(\lambda)=\Pr[\mathsf{G}_{\mathsf{forge}}=1].
\]

\paragraph{Game $\mathsf{G}_{\mathsf{replay}}$ (Cross-domain / replay misuse).}
The adversary observes valid attestations and wins if it reuses one attestation on a different $(\ObjHash,\domain,\mathsf{nonce})$ tuple and still passes relay checks. Advantage is
\[
\Adv^{\mathsf{replay}}_{\ARACE}(\lambda)=\Pr[\mathsf{G}_{\mathsf{replay}}=1].
\]

\paragraph{Game $\mathsf{G}_{\mathsf{admit}}$ (Admission bypass under policy).}
Given the admissibility policy in Section~\ref{sec:problem:admission}, the adversary wins if it outputs $(\mathsf{obj},\Attest)$ that is accepted and forwarded by an honest relay despite violating the policy (ineligible key, quota breach, or denylisted key). Advantage is
\[
\Adv^{\mathsf{admit}}_{\ARACE}(\lambda)=\Pr[\mathsf{G}_{\mathsf{admit}}=1].
\]

\paragraph{Executable interfaces for admission assumptions.}
To make the admission bound machine-checkable at the model level, we define the three terms via explicit experiments:
\begin{itemize}[nosep,leftmargin=1.5em]
  \item $\Adv^{\mathsf{meta}}(\lambda)$: advantage in $\mathsf{G}_{\mathsf{meta}}$, where the adversary outputs eligibility metadata (or update transcript) accepted by honest relays but not authorized by the metadata root/signing authority for the same epoch.
  \item $\Adv^{\mathsf{state}}(\lambda)$: probability in $\mathsf{G}_{\mathsf{state}}$ that honest relay state transitions violate the specified admission state machine (quota counter, nonce set, epoch reset) under an admissible event trace.
  \item $\Adv^{\mathsf{net}}(\lambda)$: advantage in $\mathsf{G}_{\mathsf{net}}$ that an in-flight message modification changes admission-relevant fields (key id, nonce, domain, epoch tag) while still passing endpoint integrity checks.
\end{itemize}
All three are parameterized by concrete relay implementation and network channel assumptions; they can be instantiated directly in a simulation or model-checking harness.

\begin{theorem}[Relay-path authentication and anti-replay]
Assume EUF-CMA security~\cite{gmr88} of the attestation signature scheme, collision resistance of $H$, and strict domain/nonce binding in the signed message. Then
\[
\Adv^{\mathsf{forge}}_{\ARACE}(\lambda)+\Adv^{\mathsf{replay}}_{\ARACE}(\lambda) \leq \negl(\lambda).
\]
\end{theorem}

\begin{proof}[Proof sketch]
A successful forgery gives an EUF-CMA forgery against the signature scheme. A successful replay or substitution across $(\ObjHash,\domain,\mathsf{nonce})$ tuples implies either a valid signature reuse on a distinct message (again an EUF-CMA break) or a hash collision or domain-binding failure. By a standard reduction and union bound, both advantages are negligible. Formally, $\Adv^{\mathsf{forge}}_{\ARACE}(\lambda)+\Adv^{\mathsf{replay}}_{\ARACE}(\lambda)$ is bounded by the sum of the EUF-CMA advantage of the attestation signature scheme and the collision-finding advantage against $H$, both negligible in $\lambda$.
\end{proof}

\begin{theorem}[Admission-policy soundness]
Assume (i) authenticated eligibility metadata (eligible-key set, denylist, epoch parameters) with tampering advantage $\Adv^{\mathsf{meta}}(\lambda)$, (ii) correct quota/nonce state enforcement by honest relays with inconsistency probability $\Adv^{\mathsf{state}}(\lambda)$, and (iii) network integrity against in-flight admission-field manipulation with advantage $\Adv^{\mathsf{net}}(\lambda)$. Then
\[
\Adv^{\mathsf{admit}}_{\ARACE}(\lambda)\leq
\Adv^{\mathsf{meta}}(\lambda)+
\Adv^{\mathsf{state}}(\lambda)+
\Adv^{\mathsf{net}}(\lambda).
\]
In particular, if the three terms are negligible in $\lambda$, then $\Adv^{\mathsf{admit}}_{\ARACE}(\lambda)\leq \negl(\lambda)$.
\end{theorem}

\begin{proof}[Proof sketch]
An admission bypass can occur only if at least one of the following happens: metadata is forged/stale so an ineligible key appears eligible; relay state is inconsistent so quota or nonce checks are skipped; or admission-relevant fields are modified in transit without detection. These events upper-bound $\mathsf{G}_{\mathsf{admit}}$ by a union bound, giving the stated inequality.
\end{proof}

\begin{proposition}[Inclusion safety under proof-off-path]
If the aggregated proof system is sound and consensus rejects blocks with invalid aggregated proof $\Pi$, then no object failing the validity predicate can be finalized, regardless of relay-path admission outcomes.
\end{proposition}

\begin{proof}[Proof sketch]
Relay admission affects propagation, not final validity. Finalization requires a valid $\Pi$ over included objects. Any invalid object would invalidate $\Pi$, causing consensus rejection.
\end{proof}

\subsection{Relay Integrity}
\label{sec:security:relay}
Relay nodes only forward $(\mathsf{obj}, \Attest)$. They do not need to trust submitters for \emph{validity}; they only need to verify that $\Attest$ is well-formed and bound to $\ObjHash$ and $\domain$. Invalid or malformed objects can still be forwarded; the builder will reject them when building the aggregated proof (only valid objects can be included in a valid $\Pi$). Relay nodes therefore do no validity checking, and the security of inclusion rests on the builder's proof.

\subsection{Spam and DoS}
\label{sec:security:spam}
If attestation is cheap (e.g., anyone can create a key and sign), spam of $(\mathsf{obj}, \Attest)$ could flood the network. Mitigations include: (i) attestation keys tied to stake or identity commitment (e.g., only registered or bonded keys); (ii) rate limiting per key or per commitment; (iii) small bond or fee per attestation. These are protocol-layer choices; \ARACE only requires that $\Attest$ is compact and verifiable.

\subsection{Binding and Replay}
\label{sec:security:binding}
$\Attest$ must bind to $\ObjHash$ (and $\domain$) so that an attestation for one object cannot be reused for another. A nonce or sequence number in the signed message prevents reuse of the same attestation for the same object across time or contexts. The epoch $e$ is signed and checked against the relay's current epoch, so an attestation from an earlier epoch cannot be replayed after the per-epoch nonce sets and quotas reset. Nonces are deduplicated per key, as $(\mathsf{pk}, \mathsf{nonce})$, so one key cannot pre-empt another key's nonce. Domain separation prevents cross-chain or cross-mempool replay.

\subsection{Builder Security}
\label{sec:security:builder}
The builder is assumed to produce a correct aggregated proof $\Pi$ for the set it includes. Producing $\Pi$ requires the validity witness of every included object, so \ARACE applies to object classes whose witness travels inside $\mathsf{obj}$ (for example, a transaction together with its own PQC signature) or whose validity predicate needs no external witness. Because one invalid object in $S$ prevents $\Pi$, the builder checks each candidate object before inclusion (Algorithm~\ref{alg:arace}); the per-object verification cost thus moves from relays to the builder rather than disappearing. Builder incentives and slashing (e.g., in proof-of-stake or builder markets) are outside the scope of this paper; we assume that the builder's output is verified by the consensus layer and that an invalid $\Pi$ is rejected.

\section{Structural Bandwidth Comparison}
\label{sec:bandwidth}

Let $n$ be the number of objects, $d$ the relay node degree (number of peers), and $T$ the tick interval (e.g., 0.5\,s). Assume one STARK~\cite{stark} $\approx$ 128\,KB and one attestation $\approx$ 64--256 bytes (classical, e.g., Ed25519~\cite{ed25519}) or $\approx$2.5\,KB (PQC signature, e.g., ML-DSA~\cite{mldsa}).

\paragraph{Recursive STARK propagation~\cite{vbuterin-recursive-stark-mempool}.}
Each node produces one recursive STARK per tick and sends it to each peer. Per-node proof bandwidth (outbound): $128\,\text{KB} \times d / T$. Inbound is similar, for a total of $\approx 2 \times 128\,\text{KB} \times d / T$ (e.g., $\approx$2\,MB/s outbound and $\approx$4\,MB/s in total for $d=8$, $T=0.5$\,s). Object payloads: $n$ objects $\times$ object size, as in status quo. \emph{Proof traffic on path:} $O(d/T)$ per node, constant in $n$.

\paragraph{AR-ACE (proof-off-path).}
No proof is sent on the path. Per object, each link carries $\mathsf{obj}$ plus $\Attest$. Per-node bandwidth for \emph{proof-related} data is therefore zero. For attestations: at most $n \times |\Attest|$ over time (objects and attestations flow together). If attestations are 64--256 bytes, total attestation overhead is $n \times 64$--256 bytes, which is the same order as today's per-object metadata (e.g., signatures in the non-PQC case). \emph{Proof traffic on path:} 0. Validity proof exists only at the builder (one aggregated proof in the block).

\paragraph{Reduction in proof traffic.}
Compared to \emph{proof-carrying} propagation (each object with its own 128\,KB STARK), \ARACE removes $n \times 128\,\text{KB}$ from the path. Compared to recursive STARK propagation, \ARACE removes $128\,\text{KB} \times d / T$ per node per tick from the path. In both cases, the propagation path carries no validity proofs; only the builder holds and publishes one proof. For \emph{proof-related} bandwidth, \ARACE therefore reduces on-path traffic to zero. Total authentication overhead is a different comparison: \ARACE sends $(|\Attest| + |\mathsf{pk}| + |\mathsf{nonce}|)$ bytes per object on each link, while recursive STARK propagation sends $|\pi|$ per link per tick, so \ARACE carries less authentication data only when fewer than $n^*_{\mathrm{tick}} = |\pi| / (|\Attest| + |\mathsf{pk}| + |\mathsf{nonce}|)$ objects arrive per tick (about 35 per tick for ML-DSA-44 with an inline public key and a 128\,KB proof). This statement is metric-scoped: it does not claim identical reduction for total relay traffic under adversarial invalid-object pressure or policy-control traffic.

\subsection{Parametric Compute Envelope (Analytical, Non-Empirical)}
\label{sec:bandwidth:compute-example}
The relations below are hardware-agnostic and analytical; they are not benchmark claims.

For a 10-minute analysis window, let $N$ be objects observed, $N_{\mathrm{ticks}}$ be relay ticks, $K$ be relay nodes, $c_h$ be per-tick heavy verification cost (recursive-proof path), and $c_l$ be per-object light attestation verification cost. The non-builder relay CPU ratio is
\[
\frac{\text{AR-ACE CPU}}{\text{Recursive-on-path CPU}}=\frac{N\cdot c_l}{N_{\mathrm{ticks}}\cdot c_h}.
\]
\noindent AR-ACE yields lower relay verification CPU iff
\[
\frac{c_h}{c_l}>\frac{N}{N_{\mathrm{ticks}}}.
\]
For the workload pair $(N,N_{\mathrm{ticks}})=(50{,}000,1{,}200)$, this threshold is
\[
\frac{c_h}{c_l}>41.67.
\]
Thus, any deployment where heavy-check cost exceeds light attestation-check cost by more than $41.67\times$ has lower relay-side CPU cost; the exact percentage then follows from the measured $(c_h,c_l)$ on the target implementation/hardware.

Network-wide relay CPU scales by the same ratio:
\[
\frac{K\cdot N\cdot c_l}{K\cdot N_{\mathrm{ticks}}\cdot c_h}
=\frac{N\cdot c_l}{N_{\mathrm{ticks}}\cdot c_h}.
\]

\noindent The ratio does not depend on $K$. Heavy verification moves off relay hops to inclusion time, and the actual gain depends on measured platform parameters.

\subsection{Practical Optimization Directions}
\label{sec:bandwidth:optimizations}
Several implementation-level optimizations can improve throughput and latency further:
\begin{itemize}[nosep,leftmargin=1.5em]
  \item \textbf{Batch verification of attestations.} Relays can verify many attestations together to reduce per-object verification overhead.
  \item \textbf{Two-stage propagation (hash+tag then fetch).} Nodes first relay compact metadata and fetch full objects on demand, reducing wasted bandwidth on invalid or low-priority objects.
  \item \textbf{Adaptive admission control.} Per-key quota, fee, and bond parameters can be adjusted by congestion regime to control spam pressure.
  \item \textbf{Builder-side parallel pre-check pipeline.} Builders can parallelize lightweight filtering and bucketization before aggregated proving, shortening the heavy critical path.
  \item \textbf{Profile switching (Hybrid vs Full-PQC).} Deployments can run low-overhead hybrid attestation in normal conditions and switch to full-PQC attestation under elevated threat assumptions.
\end{itemize}

\begin{center}
\textbf{Optimization directions and their primary impact metrics.} Quantitative gains are implementation- and workload-dependent; this table identifies the structural relationship only.
\vspace{0.5em}
\begin{tabularx}{0.98\linewidth}{@{}l l X@{}}
\toprule
Optimization & Primary metric & Mechanism \\
\midrule
Batch attestation verification & Relay CPU & Amortize verification overhead across multiple attestations per batch \\
Two-stage propagation (hash+tag) & Relay bandwidth & Defer full object fetch; avoid transmitting invalid or low-priority objects \\
Adaptive admission control & Adversarial overhead & Throttle ineligible or over-quota submitters at relay entry \\
Builder parallel pre-check & Builder latency & Parallelize lightweight filtering before aggregated proving \\
Hybrid-profile operation & Relay CPU + bandwidth & Use compact classical attestations when PQC threat model permits \\
\bottomrule
\end{tabularx}
\end{center}

\noindent Empirical evaluation under realistic mempool conditions is left for future work.

\section{Advantages of the ACE-GF Instantiation}
\label{sec:acegf}

Using ACE-GF to derive attestation keys yields:

\begin{itemize}[nosep,leftmargin=1.5em]
  \item \textbf{Single identity root.} The same root ($\REV$) that backs on-chain identity commitments and wallet keys derives the relay attestation key, so no separate attestation key or backup is needed.
  \item \textbf{Context isolation.} The attestation key is in a dedicated derivation context. A relayer or compromised relay-only key does not expose authorization or signing keys.
  \item \textbf{PQC readiness.} A PQC stream (e.g., ML-DSA) can be used for attestation, keeping the entire pipeline PQC-consistent. Migration to PQC attestation does not require a new identity or backup.
  \item \textbf{Determinism and auditability.} Same identity and context always yield the same attestation key. Recovery and audit (which key attested) follow from this.
\end{itemize}

Proof-off-path propagation does not require these properties; they matter when a single identity root and PQC consistency are design goals.

\section{Conclusion}
\label{sec:conclusion}

We introduced \ARACE, a proof-off-path propagation design for post-quantum mempools. By carrying only lightweight attestations on the relay path and performing a single aggregated validity proof at the builder, \ARACE removes proof traffic, and with it proof-related bandwidth, from the relay path. Security is formalized via three games (attestation forgery, replay, and admission bypass) with reduction-based arguments under standard assumptions. The design is compatible with existing validity predicates and aggregated proof systems (e.g., recursive STARKs at the builder). When attestation keys are derived with ACE-GF, they share an identity root with on-chain authorization keys and can use PQC algorithms. We leave concrete parameterization, implementation, and integration with specific mempool and builder implementations for future work.


\end{document}